\documentclass[]{fundam}
\usepackage{url} % takes care of hyperlinks, preferred over hyperref
\usepackage{ amssymb }
\usepackage{scalerel}
\def\sq{\mathbin{\scalerel*{\strut\rule[-.5ex]{2ex}{2ex}}{\cdot}}}
\usepackage[ruled,lined,linesnumbered]{algorithm2e}% provides Algorithm environment
\usepackage{graphicx}% allows for inclusion of graphic files (res)
\usepackage{tikz}
\usepackage{forest}
\usetikzlibrary{automata, positioning, arrows, arrows.meta}

\begin{document}
	
	%%%%%%%  to be filled in by copy-editor:  %%%%%%%%%%%
	%
	\setcounter{page}{1}
	\publyear{2026}
	\papernumber{0001}
	\volume{178}
	\issue{1}
	%
	%%%%%%%%%%%%%%%%%%%%%%%%%%%%%%%%%

	\title{Fuzzy Pattern Matching in Ordered Structures}

	\author{Armen Kostanyan\corresponding \\ 
		American University of Armenia\\ Yerevan, Armenia\\
		armko@aua.am
		\and Arevik Harmandayan \\
	American University of Armenia\\ Yerevan, Armenia\\ 
		arevik.harmandayan@gmail.com } 
	
	\maketitle
	
	\runninghead{A. Kostanyan, A. Harmandayan}{Fuzzy Pattern Matching in Ordered Structures}
	
	\begin{abstract}
	
	The problem of pattern matching, that is, finding all occurrences of a given pattern in a string, is one of the fundamental problems in computer science that has applications in many areas.
	
	In this paper, we consider \emph{fuzzy} patterns, defined as sequences of fuzzy properties over the basic alphabet. We first consider fuzzy pattern matching for sequences of elements of the basic alphabet and then extend the problem to partially ordered sets of nodes labeled by elements of the basic alphabet. For sequences, we seek segments that match the pattern, whereas for partially ordered structures, we seek saturated chains of nodes that match the pattern.
	
	The key concept underlying the solutions to these problems is the notion of a \emph{trajectory}, which generalizes the concept of the prefix function used in the \emph{Knuth--Morris--Pratt (KMP)} algorithm. A trajectory is processed together with the corresponding data structure, allowing the proposed algorithms to be represented as transition systems whose states are trajectories for sequences and trajectories associated with nodes for partially ordered structures. The trajectory-based approach provides a unified framework for fuzzy pattern matching in various data structures.
	
	\end{abstract}	
	
	\begin{keywords}
		Fuzzy pattern matching, Approximate string matching, KMP-string matching 
	\end{keywords}

\section{Introduction}
String matching, also referred to as pattern matching, is a fundamental problem in computer science. Given a pattern $P$ and a text $T$, the exact string-matching problem consists of finding all occurrences of $P$ in $T$. Efficient algorithms for string matching are widely used in applications such as plagiarism detection, cybersecurity, natural language processing, and bioinformatics \cite{ref1}.

In many practical situations, however, exact matching is too restrictive. A pattern may occur in a text with a small number of errors caused by misspellings, mutations, or other variations.
This motivates \emph{approximate string matching}, 
in which a text segment is considered an occurrence of a pattern if the distance between the two strings does not exceed a given threshold.
Two classical models are Hamming-distance matching, which allows a bounded number of character substitutions, and edit-distance matching, which additionally allows insertions and deletions \cite{ref2}.

A related notion is \emph{fuzzy string matching}, which generally refers to matching techniques that allow a degree of similarity or uncertainty in determining whether two strings correspond. 
Unlike exact matching and, in some formulations, distance-based approximate matching, fuzzy matching may represent the degree to which a matching condition is satisfied rather than relying solely on a binary match/mismatch criterion or a fixed number of editing operations. Fuzzy string matching has applications in entity resolution, search and information retrieval, and the automatic classification of noisy textual data \cite{ref3,ref4}. Fuzzy approaches have also been investigated in bioinformatics, where uncertainty and variability in biological sequences can make rigid character-by-character comparison inadequate \cite{ref5}.

The distinction between approximate and fuzzy matching is not always made consistently in the literature, and the two terms are sometimes used interchangeably. In this paper, we use the term \emph{fuzzy matching} in a more specific sense. We consider fuzzy patterns to be sequences of fuzzy properties associated with pattern symbols. 
Thus, a pattern does not necessarily specify an exact character for each pattern symbol; instead, each pattern symbol specifies a property that the corresponding text character is required to satisfy to a prescribed degree.
The matching problem is therefore to find all text segments whose characters satisfy the corresponding fuzzy properties to the specified degree.

Our previous works have addressed several problems in exact and fuzzy string matching. In \cite{ref6}, pattern periodicity was exploited to improve the efficiency of the preprocessing phase of the finite-automaton method and the Knuth--Morris--Pratt (KMP) algorithm. 
In \cite{ref7}, a nondeterministic transition system was constructed to describe the possible ways of processing a given text in order to find all occurrences of a fuzzy pattern.
In \cite{ref8}, an efficient algorithm was proposed for finding all occurrences of a fuzzy pattern in a text, following the principles of the KMP algorithm and using a two-dimensional prefix table. 
In \cite{ref9,ref10}, a different formulation of fuzzy patterns was considered, in which fuzzy properties were associated with text segments. The corresponding problem of finding segmentations of the text—i.e., decompositions into adjacent segments—such that the resulting segments match a fuzzy pattern to a specified degree was investigated.

This paper considers the fuzzy pattern matching problem in two settings.

The first setting extends classical pattern matching to fuzzy pattern matching, where a text segment is the search object. The second generalizes this setting by extending the search space to a finite partially ordered set of nodes, where a saturated chain of nodes serves as the search object. The solutions to both problems are based on the notion of a \emph{trajectory}, which generalizes the prefix function used in the KMP algorithm by maintaining information about the compatibility of suffixes of the processed sequence with prefixes of the fuzzy pattern.

The paper is organized as follows. Section 2 introduces the basic concepts and definitions. Section 3 presents a solution to the fuzzy pattern matching problem for \emph{linear structures}. The solution takes the form of an algorithm whose control component follows the structure of the KMP algorithm, while its operational component processes a trajectory object. Section 4 presents a solution to the fuzzy pattern matching problem for \emph{hierarchical structures} represented by finite ordered rooted trees. The control component of the algorithm traverses the tree in preorder, while its operational component computes the trajectory object at each visited node. Finally, Section 5 concludes the paper and summarizes the main results.
\section{Preliminaries}
\subsection{Lists and strings}
A \emph{list} is a sequence of elements from a given set, defined recursively as follows. The empty sequence is a list. A nonempty list consists of a distinguished element called the \emph{head} and a list called the \emph{tail}, formed by the remaining elements.

The following standard operations are defined for lists:
\begin{itemize}
	\item $cons(elem,list)$: given $elem$ and $list$, creates and returns a new list whose head is $elem$ and whose tail is $list$;
	\item $head(list)$: given a nonempty $list$, returns its head;
	\item $tail(list)$: given a nonempty $list$, returns its tail.
\end{itemize}
In what follows, we denote a list by a sequence of its elements enclosed in square brackets.
	
Let $\Sigma$ be a finite alphabet of characters. Following standard notation \cite{ref11}, we denote the set of all strings over the alphabet $\Sigma$, including the empty string $\varepsilon$, by $\Sigma^{*}$. The length of a string $x\in\Sigma^{*}$ is denoted by $|x|$. In particular, $|\varepsilon|=0$.

Given strings $x,y\in\Sigma^{*}$, we say that $y$ is a \emph{suffix} of $x$, denoted by $y\sqsupseteq x$, if $x=zy$ for some $z\in\Sigma^{*}$. We say that $y$ is a \emph{proper suffix} of $x$, denoted by $y\sqsupset x$, if $y\sqsupseteq x$ and $y\neq x$.

\subsection{Fuzzy symbols and fuzzy patterns}

Suppose $(D,\leq,0,1)$ is a linearly ordered set with the smallest element
$0$ and the largest element $1$. A \emph{fuzzy subset} $A$ of a set $U$
\cite{ref12} is defined by a membership function
$\mu_A:U\to D$ that associates with each element $x\in U$ a value
$\mu_A(x)\in D$, called the \emph{degree of membership} of $x$ in $A$.

A fuzzy subset $A$ of $U$ is usually represented in additive form as
\[
A=\sum_{x\in U}\mu_A(x)/x.
\]
We say that an element $x$ certainly belongs to $A$ if
$\mu_A(x)=1$, and that it certainly does not belong to $A$ if
$\mu_A(x)=0$. If $0 < \mu_A(x) < 1$, we say that $x$ belongs to $A$
with degree $\mu_A(x)$.

We define a \emph{fuzzy symbol} over $\Sigma$ alphabet as a fuzzy subset of
$\Sigma$ that assigns a degree of membership to each character in
$\Sigma$. 
Given a character $c\in\Sigma$ and a fuzzy symbol $\alpha$ with membership function
$\mu_\alpha$, we say that $c$ matches
$\alpha$ with degree $\mu_\alpha(c)$.
In the complexity analysis that follows, we assume, without further mention, that $\mu_\alpha(c)$ can be computed in $O(1)$ time.

Let $\Phi(\Sigma)$ denote the set of all fuzzy symbols over $\Sigma$.
We define a \emph{fuzzy pattern} as a non-empty finite sequence of elements from $\Phi(\Sigma)$ and assume that it is represented as an array $P[1..m]$ of $m\geq 1$ components.

Given a threshold $\mu\in D$, a character $c\in\Sigma$, and a fuzzy
symbol $\alpha \in \Phi(\Sigma)$, we say that $c$ $\mu$-\emph{matches}
$\alpha$, denoted by $c\sim_\mu\alpha$, if
$\mu_\alpha(c)\geq\mu$.
This definition naturally extends to strings of equal length consisting
of characters and fuzzy symbols, respectively. Namely, given a threshold
$\mu\in D$, a string $x[1..q]$ of characters from $\Sigma$, and an array
$P[1..q]$ of fuzzy symbols from $\Phi(\Sigma)$, we say that $x$
$\mu$-\emph{matches} $P$, denoted by $x\sim_\mu P$, if
\[
x[j]\sim_\mu P[j] \text{ for all } 1\leq j\leq q.
\]
\subsection{Borders and trajectories}    
  Let $P[1..m]$ be a fuzzy
  pattern, let $x\in\Sigma^*\setminus\{\varepsilon\}$ be a non-empty string with a length of no more than $m$, and let $\mu\in D$. 
  
  A string $y\in\Sigma^*$ is called a \emph{border} of $x$ if $y\sqsupset x$ and $y\sim_\mu P[1..|y|]$. In other words, the last $|y|$ characters of $x$
  $\mu$-match the first $|y|$ fuzzy symbols of $P$.
	Given the parameters $P$ and $\mu$, we denote the longest border of $x$ by $LB(x)$, without explicitly mentioning these parameters. 
  
  For every string $x \in \Sigma^{*}$, the \emph{trajectory} $\tau(x)$ is a list of strings defined by the following recursive rules:
\begin{itemize}
	\item If $x=\varepsilon$, then $\tau(x)=[\varepsilon]$;
	\item If $x\neq\varepsilon$ and $LB(x)=\varepsilon$, then $\tau(x)=[\varepsilon]$;
	\item If $x\neq\varepsilon$ and $LB(x)\neq\varepsilon$, then
	$\tau(x)=cons(LB(x),\tau(LB(x)))$.	
\end{itemize}

In addition to the standard list operations, we define the following operations on trajectories.
\begin{itemize}
	\item $length(\tau)$: given a trajectory $\tau$, returns the number of strings contained in $\tau$;
	\item $width(\tau)$: given a trajectory $\tau$, returns the value $|head(\tau)|$.

We have
\[
1 \leq length(\tau) \leq m+1, \qquad 0 \leq width(\tau) \leq m.
\]
\end{itemize}
In the subsequent analysis, we assume that these operations run in $O(1)$ time by keeping track of the corresponding values.	

We call the trajectory $[\varepsilon]$ \emph{empty} and denote it by $\tau_0$. Note that
$length(\tau_0)=1, width(\tau_0)=0$.  

\begin{example} \label{ex1}	
	Let us choose
	\[
	\Sigma=\{1,2,3,4,5\},\qquad
	D=\{0.0,0.25,0.5,0.75,1.0\},
	\]
	and define the fuzzy symbols $S$ (small), $M$ (medium), and $L$ (large)
	as follows:
	\[
	S=1.0/1+0.75/2+0.5/3+0.25/4+0.0/5,
	\]
	\[
	M=0.0/1+0.75/2+1.0/3+0.75/4+0.0/5,
	\]
	\[
	L=0.0/1+0.25/2+0.5/3+0.75/4+1.0/5.
	\]
	
	For the fuzzy pattern $P=SMSLMM$, the string $x=522142$, and the
	threshold $\mu=0.75$, we obtain
	\[
	\tau(x)= [22142, 142, 2, \varepsilon];
	\]
	\[length(\tau(x))=4, \qquad	width(\tau(x))=5.
	\]
\end{example}

\section{Fuzzy pattern matching in linear structures}

In this section, we consider the problem of fuzzy pattern matching in a string of characters from the basic alphabet. We say that a fuzzy pattern occurs in such a string if there is a segment of the same length as the pattern that $\mu$-matches the pattern.

The problem of finding all occurrences of a fuzzy pattern in a string generalizes the classical string-matching problem and was solved in \cite{ref8} using a dynamic programming method based on a two-dimensional prefix table, which generalizes the prefix-function array used in the KMP algorithm. In this section, we propose a new algorithm for solving this problem based on the \emph{trajectory algebra}, with the same asymptotic time complexity and a lower asymptotic space complexity.

\subsection{Fuzzy \emph{L}-matching problem}

Let $X[1..n]$ be a string of length $n$ over the alphabet $\Sigma$, let
$P[1..m]$ be a fuzzy pattern of length $m$, and let
$\mu\in D$ be a threshold.

We define a \emph{$\mu$-occurrence} of the fuzzy pattern $P$ in $X$ as
a segment $X[i-m+1..i], m \leq i \leq n$, such that
\[
X[i-m+1..i]\sim_\mu P[1..m].
\]
Accordingly, we define the problem of fuzzy pattern matching as finding
all $\mu$-occurrences of $P$ in $X$, which are identified by their end
positions in $X$.

We call this problem the \emph{fuzzy L-matching problem}.
\newpage
\subsection{Trajectory algebra}

We define the following operations on trajectories. Below, we assume that
$\tau$ is a trajectory with $head(\tau)=x$.

\begin{itemize}
	\item
	$isPromoted(\tau,c)$, $c\in\Sigma$.
	Given that $|x|<m$ and $x\sim_\mu P[1..|x|]$, checks whether
	$xc$ $\mu$-matches the corresponding prefix of $P$, i.e.,
	\[
	xc\sim_\mu P[1..|x|+1].
	\]
	This check reduces to testing $c\sim_\mu P[|x|+1]$ and thus
	has complexity $O(1)$.
	
	\item
	$promote(\tau,c)$, $c\in\Sigma$.
	Given that the precondition of the $isPromoted$ operation is satisfied and $isPromoted(\tau,c)=true$, the operation returns the trajectory formed by the strings in the following set, ordered by decreasing length:
	\[
	 \tau'=
	\bigl\{
	x'c \mid x'\in\tau,\;
	c\sim_\mu P[|x'|+1]
	\bigr\} \cup \{ \varepsilon \}.
	\]
	Note that $\tau'$ satisfies the trajectory requirements; that is,
	each subsequent element is the longest border of its preceding
	element, and $head(\tau')=xc$.
	The returned trajectory satisfies 
	\[length(\tau') \leq length(\tau)+1, width(\tau') = width(\tau)+1.\] 
	The operation has complexity $O(m)$.
	
	\item
	$rollback(\tau)$.
	Given that $|x|>0$, returns the trajectory
	$\tau'=tail(\tau)$.	The returned trajectory satisfies 
	\[length(\tau') \leq length(\tau)-1, width(\tau') \leq width(\tau)-1. \]  
	The operation has complexity $O(1)$.
\end{itemize}

We call the set of trajectories, together with their basic operations extended by the operations in the set $\{isPromoted, promote, rollback\}$, the \emph{trajectory algebra}, and denote it by $\Xi$. In this section, we associate the following meanings with the elements of $\Xi$: 
If $\tau \in \Xi$, then $\tau$ consists of those suffixes $x'$ of the portion of the text processed so far for which $x'\sim_\mu P[1..|x'|]$. As the text is processed, the trajectory changes accordingly.
When $width(\tau)=m$, a $\mu$-occurrence of the fuzzy pattern is detected, and the trajectory is reset.

The notion of a \emph{trajectory} generalizes the role of the prefix function used in the KMP algorithm. In the classical KMP algorithm for exact pattern matching \cite{ref13,ref14}, the prefix function array of length $m$ is computed during the preprocessing phase based solely on the pattern. In the case of a fuzzy pattern, the text is processed together with the pattern, and the matching history is continuously maintained and updated in a trajectory of length at most $m$.
 \subsection{Fuzzy \emph{L}-matching algorithm}
 Figure \ref{fig:1} presents a solution to the fuzzy pattern matching problem for linear structures based on the trajectory algebra.
$\newline$

\begin{algorithm}[H] \caption{$FL-Matching //$ \text{Prints all valid shifts}}     
	
	\KwIn{
		
		\text{   } Fuzzy pattern $P=P[1..m]$, string $X=X[1..n]$, and threshold $\mu \in D$
	}
	\KwOut{
		All $\mu$ - occurrences of $P$ in $X$  
	}
	
	$\tau=\tau_0$        //current trajectory
	
	\For{$i = 1$ \textbf{to} $n$}{ 
		\While{
			$length(\tau)>1$ and  not $isPromoted(\tau, X[i])$ 			
		}
		{ 
			$\tau = rollback( \tau ) $
		}
		\If{$isPromoted(\tau, X[i])$}{
			$\tau = promote( \tau, X[i] )$ }
		\If{$width(\tau) == m$ //entire pattern matched}                 {
			$Print( i ) $ 
			
			$\tau = rollback( \tau )$
		}
		
	}
\end{algorithm} 
\begin{figure}[!h] 
	\caption{The $FL-Matching$ algorithm.}
	\label{fig:1}
\end{figure}

\subsection{Fuzzy linear matching in terms of a transition system} \label{sec 3.3}
The $FL-Matching$ algorithm can be viewed as a process in the trajectory algebra driven by the successive characters of the input sequence $X[1..n]$.

More precisely, consider the labeled transition system
\[
T_1=(S,s_0,F,\Sigma,\Delta),
\]
constructed for the fuzzy pattern $P$ as follows:

\begin{itemize}
	
	\item[$\sq$]
	$S=\Xi$ is the set of \emph{states};
	\item[$\sq$]
	$s_0=\tau_0$ is the initial state;
	\item[$\sq$]
	$F\subseteq S$ is the set of final states, consisting of trajectories of width $m$. Upon reaching a final state, an occurrence of the fuzzy pattern in the text is detected;
	\item[$\sq$]
	$\Sigma$ is the input alphabet;
	\item[$\sq$]
	$\Delta:S\times(\Sigma\cup\{\epsilon\})\to S$ is the transition
	function defined by the following rules. In what follows, $\epsilon$ denotes an $\epsilon$-transition that consumes no input symbol, whereas $\varepsilon$, as defined above, denotes the empty string.	
	 
	$L1.$ \emph{Trajectory reduction:} 
	\[
	\frac
	{
		length(\tau) > 1,  
		\neg isPromoted(\tau, c)
	}
	{ \tau \xrightarrow{c} \tau', 
		\tau' = rollback(\tau)
	},
	\] 
	
	$L2.$ \emph{Trajectory update:}
	\[
	\frac{width(\tau) < m, 
		isPromoted(\tau, c) 
	}
	{\tau \xrightarrow{c} \tau', 
		\tau' = promote(\tau, c)
	}, 
	\]
	
	$L3.$ \emph{Trajectory reset:}
	\[
	\frac{ width(\tau) = m }
	{ \tau \xrightarrow{\epsilon} \tau',  
		\tau'= rollback(\tau) 
	}. 
	\]
\end{itemize}

\begin{example} \label{ex2}
	
	Let us represent the execution of the $FL-Matching$ algorithm in the transition system $T_{1}$, assuming that
	 \[
	 P[1..4] = SMSL, 
	 \]
	 \[
	 X[1..8] = 13231425,
	 \] 
	 \[
	 \mu=0.75,
	 \]
	and that the fuzzy symbols $S$, $M$, and $L$ are defined as in Example $\ref{ex1}$.
	 
	The execution is as follows:
	\[
	\tau_{0} = [\varepsilon]
	\xrightarrow{132} 
	//\text{applying rule } L2 \text{ three times}
	\]
	\[
	\tau_{1} = [132, 2, \varepsilon]
	\xrightarrow{3} 
	//\text{applying rule } L1 
	\]
	\[
	\tau_{2} =[2, \varepsilon]
	\xrightarrow{314}  
	//\text{applying rule } L2 \text{ three times,} \textit{ matched}  
	\]
	\[
	\tau_{3}^{*} = [2314, 14, \varepsilon]
	\xrightarrow{\epsilon}
	//\text{applying rule } L3
	\]
	\[	
	\tau_{4} = [14, \varepsilon]
	\xrightarrow{25}
	//\text{applying rule } L2 \text{ twice,} \textit{ matched}    
	\] 
	\[
	\tau_{5}^{*} = [1425, \varepsilon]
	\xrightarrow{\epsilon} 
	//\text{applying rule } L3  
	\]
	\[
	\tau_{6} = [\varepsilon].
	\]	

	Thus, we have found two occurrences of the fuzzy pattern $P=SMSL$ in $X$, namely, the segments $2314$ and $1425$.
	 
\end{example}

\subsection{Justification and analysis}
Let us call the point immediately before the \textbf{if} statement in line $9$ of the $FL-Matching$ algorithm an \emph{inflection point}. At this point, the algorithm checks whether $width(\tau)=m$; if so, a $\mu$-occurrence of the pattern is detected and the trajectory is rolled back.
\begin{lemma}[Trajectory inflection lemma]\label{lem1}
	 Let $\tau_i$ be the trajectory at the inflection point in the $i$-th iteration of the $FL-Matching$ algorithm, $1 \leq i \leq n$.
	 Then $\tau_i$ consists of all borders of the string $X[1..i]$.
\end{lemma}
\begin{proof}
We prove the lemma by induction on $i$. The base case $i=1$ is immediate. Assume that the assertion holds for $i-1$, where $2 \leq i \leq n$, and prove that it also holds for $i$.
   
 First, note that $\tau_i$ is obtained from $\tau_{i-1}$ by applying the $rollback$ operation several times, followed by the $promote$ operation at most once. After all rollback operations have been applied, we obtain a trajectory whose elements are all borders of the string $X[1..i-1]$. The subsequent execution of the $promote$ operation in line $7$ constructs a trajectory whose elements are borders of $X[1..i]$. Hence, every element of $\tau_{i}$ is a border of $X[1..i]$.

To prove the converse, suppose that there exists a non-empty border $y$ of $X[1..i]$ that is not an element of $\tau_i$.

Then, $y=y'X[i]$, where $y'$ is a border of $X[1..i-1]$ and $X[i]\sim_\mu P[|y|]$. By the induction hypothesis, $y'$ is an element of $\tau_{i-1}$.
It cannot be removed from this trajectory when executing the $rollback$ operation in line $11$ of the $(i-1)$-th iteration, since $|y^{'}|<m$, nor when executing the $rollback$ operation in line $4$ of the $i$-th iteration, since 
\[
         X[i] \sim_\mu P[|y|].
\]
As a result, $y=y'c$ must have been added to $\tau_i$ after the execution of the $promote$ operation in line $7$, which yields a contradiction.
\end{proof}
\begin{corollary}\label{cor1}  
	  Let $\tau_i$ be the trajectory at the inflection point in the $i$-th iteration of the $FL-Matching$ algorithm, $1 \leq i \leq n$. Then 
	  \[
	  	head(\tau_i)=LB(X[1..i]).
	  \]
\end{corollary}
\begin{remark}\label{rem1}  
Note that Lemma ~\ref{lem1} can be viewed as stating that $head(\tau_i)=LB(X[1..i])$ is an \emph{invariant} of the algorithm at the inflection point.
\end{remark}

\begin{theorem}\label{th1}
	The \emph{FL-Matching} algorithm finds all and only the $\mu$-occurrences of the fuzzy pattern $P$ in the string $X$.
\end{theorem}

\begin{proof}
	Note that whenever $width(\tau)=m$ at the inflection point, the algorithm detects a $\mu$-occurrence of $P$ in $X$. Thus, the algorithm reports only $\mu$-occurrences of $P$ in $X$.  
	
	Conversely, suppose that
	\[
	X[i-m+1..i]\sim_{\mu}P[1..m], m \leq i \leq n,
	\]
	is a $\mu$-occurrence of $P$ in $X$. 
	Since $X[i-m+1..i]$ is a suffix of $X[1..i]$ of length $m$ that $\mu$-matches $P$, it is the longest border of $X[1..i]$.
	By Corollary ~\ref{cor1},
	\[
	 head(\tau_i)=LB(X[1..i])=X[i-m+1..i],
	\]
	and hence $width(\tau_i)=|head(\tau_i)|=m$ at the inflection point of the $i$-th iteration of the algorithm. Therefore, the algorithm prints the specified $\mu$-occurrence in line $10$. 	
\end{proof}
\title{ANALYSIS}

Let us use the \emph{potential method} \cite{ref15} to estimate the complexity of the \textbf{for} loop in lines 2--13. Define the potential at the beginning of each iteration as $length(\tau)$. Initially, the potential is $0$, and it never becomes negative.

Note that the \textbf{while} loop in lines 3--5 has $O(1)$ amortized complexity: its actual cost is compensated by a decrease in the potential.

The first \textbf{if} statement in lines 6--8 has $O(m)$ amortized complexity: its actual cost is $O(m)$, while the potential increases by at most $1$.

The second \textbf{if} statement in lines 9--12 has $O(1)$ amortized complexity: its actual cost is $O(1)$, and the potential decreases by at least $1$.

Therefore, the amortized cost of one execution of the body of the \textbf{for} loop is $O(m)$. Since the loop is executed $O(n)$ times, the total time complexity of the \emph{FL-Matching} algorithm is $O(mn)$.

The space complexity of the algorithm is bounded by the maximum capacity of a trajectory $\tau$, which is at most
\[
length(\tau) \cdot width(\tau) = O(m^2).
\]

\section{Fuzzy pattern matching in hierarchical structures}
In this section, we consider the problem of fuzzy pattern matching in a hierarchical structure represented as an ordered rooted tree whose nodes are labeled with characters from the basic alphabet. We say that a fuzzy pattern occurs in such a structure if there is a descending chain of adjacent nodes of the same length as the pattern whose sequence of labels $\mu$-matches the fuzzy pattern.

We propose an algorithm for finding all occurrences of a fuzzy pattern in a tree by traversing it in preorder and tracking trajectories at the visited nodes.
\subsection{Fuzzy \emph{H}-matching problem}

Let $G$ be a finite partially ordered set of nodes with a smallest element $\Lambda$, which is isomorphic to a tree whose root node is $\Lambda$. We assume that this tree is equipped with an ordering and refer to $G$ as an ordered rooted tree. For the tree $G$, the terms "parent" and "child" are used in their usual sense.  

We denote the set of child nodes of a node $g\in G$ by $succ(g)$;
this set is empty for leaf nodes.
We define a \emph{descending chain} of nodes as a sequence
$g_1,\ldots,g_k$ of $k\geq 1$ of nodes such that
$g_i\in succ(g_{i-1}), 2\leq i\leq k$,
and call $k$ the length of the chain.
We represent a chain of length $k$ as an array
$g[1..k]$ of nodes.
A descending chain in the derived tree corresponds to a saturated chain in the underlying partially ordered set. Accordingly, we use the terms \emph{descending chain} in the tree and \emph{saturated chain} in the partially ordered set interchangeably to refer to the same sequence of nodes.

Let $\Sigma$ be the basic alphabet. A mapping
$\xi:G\rightarrow\Sigma$ that assigns a character from $\Sigma$ to
each node of $G$ is called a \emph{labeling} of $G$.
We call the triple $\widehat{G}=\langle G,\Sigma,\xi\rangle$ 
a \emph{hierarchical structure} over $G$.

Given a fuzzy pattern $P[1..m]$, a chain
$g=g[1..m]$ of length $m$ in a hierarchical structure
$\widehat{G}$, and a threshold $\mu\in D$, we say that $g$
$\mu$-matches $P$, denoted by
$g\sim_{\mu}P$, if
\[
\xi(g[j])\sim_{\mu}P[j],\qquad 1\leq j\leq m.
\]
We define the problem of finding a fuzzy pattern $P$ in a hierarchical structure $\widehat{G}$ as the problem of finding all descending chains in $G$ that $\mu$-match $P$, with each occurrence identified by its end node in $G$.
We call this problem the \emph{fuzzy H-matching problem}.

Note that the fuzzy $L$-matching problem is a special case of the fuzzy
$H$-matching problem if the string $X[1..n]$ is viewed as a hierarchical
structure whose underlying tree consists of a single descending chain of length $n$.
\subsection{Fuzzy \emph{H}-matching algorithm}
We represent our solution to the fuzzy $H$-matching problem as
a \emph{discrete converter}. Its control component, represented by the
\emph{FH-Matching} algorithm, traverses the tree $G$ in \emph{preorder}
using a stack that stores nodes together with trajectories of the parent nodes. Its operational component, represented by the $visitNode$ subroutine, promotes the trajectory associated with the parent node to the current node.

Detailed descriptions of these algorithms are presented in
Figures~\ref{fig:3} and~\ref{fig:4}.

\begin{algorithm}[H] \caption{$FH-Matching //$ \text{Control component}}     
	
	\KwIn{
		\text{   } Fuzzy pattern $P=P[1..m]$, hierarchical structure $\widehat G = (G, \Sigma, \xi)$, threshold $\mu \in D$
	}
	\KwOut{
		All $\mu$ - occurrences of $P$ in $\widehat G$ 
	}
	
	$stack.push( \langle\Lambda, \tau_{0} \rangle )$ // stack initialization
	
	\While{ stack is non-empty }
	{
		$\langle g, \tau \rangle = stack.top(), stack.pop()$
		
		$\tau = visitNode(g, \tau)$
		
		let $succ(g) = \langle g_{1}, ..., g_{h} \rangle, h \geq 0$
				 
		\For{$i = h$ \textbf{downto} $1$} 
		{
				$stack.push(\langle g_{i}, \tau \rangle )$    
	    }		
	} %while
\end{algorithm} 
\begin{figure}[!h] 
	\caption{The $FH-Matching$ algorithm.}
	\label{fig:3}
\end{figure}

\begin{algorithm}[H] \caption{$visitNode //$ \text{Operational component}}     
	
	\KwIn{	
		\text{   } Current node $g \in G$ labeled by $c=\xi(g)\in \Sigma$, trajectory $\tau$ for the parent of node $g$
	}
	\KwOut{
		Trajectory for node $g$  
	}	 
	\While{
			$length(\tau) > 1$ and not $isPromoted(\tau, c)$ 			
	}
	{ 
			$\tau = rollback( \tau ) $
	}
	
	\If{$isPromoted(\tau, c)$}{
			$\tau = promote( \tau, c)$ }
	
	\If
	{
		$width(\tau) == m$
	}                
	{ 	
		$Print( g )$  //an occurrence of $P$ in $\widehat{G}$ is found 
		
		$\tau = rollback( \tau )$
	}
	
	\Return {$\tau$}		

\end{algorithm} 
\begin{figure}[!h] 
	\caption{The $visitNode$ procedure.}
	\label{fig:4}
\end{figure}

\subsection{Events}
When working with hierarchical structures, it is necessary to complement the concept of a trajectory by associating it with a node. We call such a pair an \emph{event}. Processing an event involves two cases: the trajectory stored in the event may be that associated with the parent node or that associated with the node of the event. In the first case, we say that the event is \emph{unrealized}, and in the second, that it is \emph{realized}.

Formally, an event is defined as a pair $e=\langle g,\tau\rangle$, where $g\in G$ and $\tau\in\Xi$.
However, during the execution of the \emph{FH-Matching} algorithm, there are two possibilities for the trajectory associated with a given node $g$. To clarify this point, assume that

\begin{itemize}
\item
	$g\in G$ is a node,
\item
    $path(g)=g_1\ldots g_k$, where $g_1=\Lambda$ and $g_k=g$, is the unique path from the root to $g$,
\item
	$x=\xi(path(g))=\xi(g_1)\cdots\xi(g_k)\in\Sigma^{*}$ is the string of labels along this path,
\item
	$\tau$ is the trajectory of $x$.
\end{itemize}
Then, an event associated with node $g$ contains either $\tau$, in which case it is called \emph{realized}, or $rollback(\tau)$, provided that $length(\tau)>1$, in which case it is called \emph{unrealized}.

We denote by $Events^{-}$ the set of all unrealized events,
by $Events^{+}$ the set of all realized events, and by
\[
Events=Events^{-}\cup Events^{+}
\]the set of all events.
We distinguish the initial event
\[
e_{0}=\langle\Lambda,\tau_{0}\rangle\in Events^{-} 
\]
and the set
\[F=
\left\{
e=\langle g,\tau\rangle\in Events^{+}
\mid width(\tau)=m
\right\}
\]
of final events.

The execution of the \emph{FH-Matching} algorithm consists of sequentially extracting the next event from the stack of unrealized events, realizing it using the $visitNode$ procedure, and adding the resulting unrealized events corresponding to the child nodes, if any, to the stack.

For convenience, we represent the $visitNode$ procedure as a mapping
\[
visitNode: Events^{-}\rightarrow Events^{+},
\]
which maps an unrealized event
$e=\langle g,rollback(\tau) \rangle$ to the realized event $e'=\langle g,\tau \rangle$.

\subsection{Fuzzy hierarchical matching in terms of a pushdown transition system}

From a theoretical point of view, the execution of the
$FH-Matching$ algorithm can be represented as a process in the space
of events, using a stack of unrealized events and driven by the
sequence of node labels encountered during the traversal.

Using standard notation for pushdown automata \cite{ref11}, let us
consider the labeled pushdown transition system
\[
T_{2}=(S,s_{0},F,\Sigma,\Gamma,Z_{0},\Delta),
\]
where
\begin{itemize}
	\item[$\sq$]
	$S=Events\cup\{nil\}$ is the set of states, consisting of the set of events supplemented by a \emph{dummy} state $nil$;
	
	\item[$\sq$]
	$s_{0}=e_{0}$ is the initial state;
	
	\item[$\sq$]
	$F$ is the set of final states and coincides with the set of final events;
	
	\item[$\sq$]
	$\Sigma$ is the input alphabet;
	
	\item[$\sq$]
	$\Gamma=Events^{-}$ is the stack alphabet;
	
	\item[$\sq$]
	$Z_{0}=e_{0}$ is the initial stack symbol;
	
	\item[$\sq$]
	$\Delta:S\times(\Sigma\cup\{\epsilon\})\times(\Gamma \cup\{\varepsilon\})
	\rightarrow S\times\Gamma^{*}$ is the transition function defined
	by the following rules:
	
	$H1.$ \emph{Event realization:}
	\[
	\frac
	{ 
	    e = \langle g, \tau \rangle\in Events^{-}, e=Z
	}
	{
		e \xrightarrow{\xi(g), Z/\varepsilon} e',
		e' = visitNode(e)
	},
	\] 
	$H2.$ \emph{Pushing successors into the stack:}
	\[
	\frac
	{
	  e=\langle g, \tau \rangle \in Events^{+},
	  succ(g) = \langle g_{1}, ..., g_{h}\rangle,
	  h \geq 0	
	}
	{
	 	e \xrightarrow{\epsilon, 
	 		\varepsilon/Z_{1} ... Z_{h} } 
 		nil,
	 	Z_{1}=\langle g_1, \tau \rangle, 
	 	..., 
	 	Z_{h}=\langle g_{h}, \tau \rangle	
	},
	\]
	
	$H3.$ \emph{Popping the next unrealized event from the stack:}
	\[
	\frac
	{
		Z\in Events^{-}
	}
	{
		 nil 
		 \xrightarrow{\epsilon, Z/\varepsilon} 
		 Z.
	}
	\]
\end{itemize} 

We define a \emph{configuration} of our transition system as a pair 
\[\langle e;stack\rangle,
\] 
where $e\in S$ and $stack$ denotes the contents of the stack, listed from top to bottom, at the time $e$ is processed.
Among the configurations, we distinguish the initial configuration $\langle e_0;[Z_0]\rangle$ and the final configurations, whose events are final events.

A sequence of configurations starting from the initial configuration and following the transition function is called a \emph{process} in the transition system. Such a process outputs a node whenever it visits a final state. We call a process \emph{terminated} if its last configuration is
$\langle nil; stack=[] \rangle$.
\begin{example} \label{ex2.3}
	We represent the execution of the \emph{FH-Matching} algorithm for the fuzzy pattern $P[1..2]=SM$, where $S$ and $M$ are defined as in Example~\ref{ex1}, with matching parameter $\mu=0.75$, on the hierarchical structure $\widehat{G}$ shown in Figure ~\ref{fig:academic_tree}, as a terminated process in the transition system $T_2$.
In the figure, slash notation is used to denote a node together with its label. 
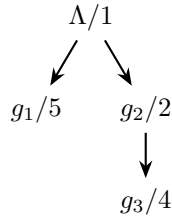
\begin{figure}[!h]
\centering
\begin{forest}
for tree={
edge={-Stealth, thick},
l sep=0.6cm, 
s sep=0.5cm,
math content,
font=\small
}
[\Lambda / 1
[g_1 / 5]
[g_2 / 2
[g_3 / 4]
]
]
\end{forest}
\caption{A hierarchical structure.}
\label{fig:academic_tree}
\end{figure}

The execution is as follows:
	\[
	\langle e_{0}; stack = [e_{0}]
	\rangle
	\xrightarrow{1}  //\text{applying rule } H1
	\] 
	\[
	\langle
	e_{1} =\langle\Lambda, \tau_{1} = [1,\varepsilon] \rangle \rangle; stack=[]
	\rangle 
	\xrightarrow{\epsilon}
	//\text{applying rule } H2
	\]
	\[
	\langle
	nil; stack = [e_{2} = \langle g_{1}, \tau_{1} \rangle, e_{3} = \langle g_{2}, \tau_{1}\rangle ] 
	\rangle
	\xrightarrow{\epsilon}
	//\text{applying rule } H3 
	\]
	\[
	\langle
	e_{2}; stack = [e_{2}, e_{3}] 
	\rangle
	\xrightarrow{5}
	//\text{applying rule } H1 
	\]	
	\[
	\langle
	e_{4}=\langle g_{1}, \tau_{0}
	\rangle;
     stack = [e_{3}] 
	\rangle
	\xrightarrow{\epsilon} 
	//\text{applying rule } H2
	\]
	\[
	\langle
	nil; stack=[e_3]
	\rangle 
	\xrightarrow{\epsilon} 
	//\text{applying rule } H3
	\]
	\[
	\langle
	e_{3}; stack=[e_{3}]
	\rangle 
	\xrightarrow{2}
	//\text{applying rule } H1 
	\]		
	\[
	\langle
	e_{5}^{*} = \langle g_{2}, \tau_{2} = [12, 2, \varepsilon] \rangle \rangle; stack=[] 
	\rangle
	\xrightarrow{\epsilon}
	//\text{applying rule } H2 
	\]
	\[
	\langle
	nil; stack = [e_{6} = \langle g_{3}, \tau_{2} \rangle ]
	\rangle 
	\xrightarrow{\epsilon}
	//\text{applying rule } H3 
	\]
	\[
	\langle
	e_{6}; stack = [e_{6}]
	\rangle 
	\xrightarrow{4}
	//\text{applying rule } H1 
	\]
	\[
	\langle
	e_{7}^{*}= \langle g_{3}, \tau_{3} = [24, \varepsilon\ \rangle \rangle; stack=[]
	\rangle 
	\xrightarrow{\epsilon} 
	//\text{applying rule } H2
	\]
	\[
	 \langle
	 nil; stack=[]
	 \rangle.  
	\]
	
	This process outputs the node $g_{2}$ upon reaching state $e_{5}$
	and the node $g_{3}$ upon reaching state $e_{7}$.
	Thus, the hierarchical structure $\widehat{G}$ contains two
	occurrences of the fuzzy pattern $P=SM$, ending at nodes $g_{2}$ and
	$g_{3}$, respectively. These occurrences correspond to the following node chains:
	\[
		(\Lambda/1)\,(g_{2}/2) \qquad \text{ and } \qquad (g_{2}/2)\,(g_{3}/4).
	\]
\end{example}

\subsection{Justification and analysis}

\begin{theorem}\label{th2}
	The \emph{FH-Matching} algorithm finds all and only the $\mu$-occurrences of the fuzzy pattern $P$ in the hierarchical structure $\widehat{G}$.
\end{theorem}
\begin{proof} 
		The result is consistent with the corresponding result for linear
		structures (Theorem~\ref{th1}). Since the preorder traversal of $G$ visits every node exactly once, every possible end node of a $\mu$-occurrence is considered by the algorithm. Moreover, after a node $g\in G$ is visited by the $visitNode$ procedure, the head of the resulting trajectory is exactly the longest suffix of $x=\xi(path(g))$ that $\mu$-matches the corresponding prefix of $P$. In particular, when the trajectory has width $m$, the corresponding chain is a $\mu$-occurrence of $P$, and every such occurrence is detected when its end node is visited.		
\end{proof}
\title{ANALYSIS}

During the preorder traversal of the tree $G$, every node is visited exactly once, and processing each node takes $O(m)$ time. Therefore, the time complexity of the \emph{FH-Matching} algorithm is $O(mN)$, where $N$ is the number of nodes in $G$.

The space complexity of the algorithm is determined by the maximum size of the stack. Since the stack contains $O(H)$ elements, where $H$ is the height of the tree, and each stack element requires $O(m^{2})$ space, the space complexity of the \emph{FH-Matching} algorithm is $O(m^{2}H)$.

\section{Conclusion}
The paper considers the problem of fuzzy pattern matching in linear and hierarchical structures. A fuzzy pattern is defined as a finite sequence of fuzzy properties associated with pattern symbols. A segment of the text over the basic alphabet matches a fuzzy pattern if they are of the same length and each character in the segment satisfies the fuzzy property associated with the corresponding pattern symbol to at least a prescribed degree.

The following items summarize our main contributions to the fuzzy pattern matching problem in linear and hierarchical structures.

\begin{itemize}
	\item
	
	\emph{Matching in linear structures.} 
	The problem of finding a fuzzy pattern in a linear structure is solved by an algorithm whose control structure is similar to that of the KMP algorithm but which processes trajectories rather than integers. The notion of a \emph{trajectory} generalizes the role of the prefix-function array underlying the KMP algorithm in the following sense. The prefix-function array is determined solely by the pattern and can therefore be constructed during the preprocessing stage. In contrast, a trajectory depends on both the pattern and the text being processed and is constructed dynamically during the execution of the algorithm. The algorithm is represented as a labeled transition system \cite{ref16} whose states are trajectories.
	
	The time complexity of the algorithm is $O(mn)$, and its space complexity is $O(m^{2})$, where $m$ and $n$ denote the lengths of the pattern and the sequence, respectively.
	
	\item
	
	\emph{Matching in hierarchical structures.} The fuzzy matching problem is extended to hierarchical structures represented by finite ordered rooted trees. In this setting, finding a fuzzy pattern amounts to finding all descending chains of nodes whose sequences of labels match the fuzzy pattern. The problem is solved by traversing the tree in preorder and tracking trajectories associated with the visited nodes. The algorithm is represented as a labeled pushdown transition system \cite{ref16} whose states are events, each consisting of a tree node together with its associated trajectory.
	
	The time complexity of the algorithm is $O(mN)$, and its space complexity is $O(m^{2}H)$, where $m$ denotes the length of the pattern, and $N$ and $H$ denote the number of nodes and the height of the tree, respectively.
\end{itemize}

\nocite{*}
\bibliographystyle{fundam}
\bibliography{citations}

\end{document}